\documentclass[runningheads]{llncs}
\usepackage{adjustbox}
\usepackage{amsmath}
\usepackage{amssymb}
\usepackage[noabbrev]{cleveref}
\usepackage{float}
\usepackage{forest}
\usepackage{mathtools}
\usepackage{doc}
\usepackage{amsfonts}
\usepackage{tikz}

\newcommand{\nol}{\overline}
\DeclarePairedDelimiter{\mset}{\{\kern-0.75ex\{}{\}\kern-0.75ex\}}
\DeclarePairedDelimiter{\seq}{(}{)}
\DeclarePairedDelimiter{\set}{\{}{\}}

\newcommand\Top{\protect\mathpalette{\protect\rlapkern}{\top}}
\newcommand\Bot{\protect\mathpalette{\protect\rlapkern}{\bot}}
\newcommand\Path{\protect\mathpalette{\protect\rlapkern}{P}}
\def\rlapkern#1#2{\rlap{$#1#2$}\mkern2mu{#1#2}}

\newcommand{\mac}{\ensuremath{\mathbf{MAC}}}
\newcommand{\macLang}{\ensuremath{\mathcal{MAC}}}
\newcommand{\macAx}{\ensuremath{\mathfrak{MAC}}}
\newcommand{\ac}{\ensuremath{\mathbf{\rm AC}}}
\newcommand{\acAx}{\ensuremath{\mathfrak{AC}}}

\newcommand{\ite}[3]{\ensuremath{(#1 ? #2\!:\!#3)}}
\newcommand{\cond}[2]{\ensuremath{(#1\! \mid\! #2)}}
\newcommand{\expect}{E}
\newcommand{\believes}{\ensuremath{\Box}}

\newcommand{\imp}{\ensuremath{\rightarrow}}

\title{A modal aleatoric calculus for probabilistic reasoning: extended version}

\author{
  Tim French\inst{1} \and Andrew Gozzard\inst{1} \and Mark Reynolds\inst{1} 
}

\institute{
  The University of Western Australia,
  Perth, Western Australia\\
  \email{[tim.french,mark.reynolds]@uwa.edu.au}\ \email{andrew.gozzard@research.uwa.edu.au}
 }

\authorrunning{French, Gozzard, Reynolds}

\titlerunning{A Modal Aleatoric Calculus for Probabilistic Reasoning}

\begin{document}

\maketitle

\begin{abstract}
  We consider multi-agent systems where agents actions and beliefs are determined aleatorically, or ``by the throw of dice''.
  This system consists of possible worlds that assign distributions to independent random variables, 
  and agents who assign probabilities to these possible worlds.
  We present a novel syntax and semantics for such system, 
  and show that they generalise Modal Logic. 
  We also give a sound and complete calculus for reasoning in the base semantics, 
  and a sound calculus for the full modal semantics, that we conjecture to be complete.  
  Finally we discuss some application to reasoning about game playing agents.
\end{abstract}

\section{Introduction}
\label{sect:introduction}

This paper proposes a probabilistic generalisation of modal logic for reasoning about probabilistic multi-agent systems.
There has been substantial work in this direction before \cite{kooi.ll:2011,baltagetal.tark:2009,jfaketal.prob:2009}. 
However, here, rather than extending a propositional modal logic with the capability to represent and reason about probabilities,
we revise all logical operators so that they are interpreted probabilistically.
Thus we differentiate between {\em reasoning about probabilities} and {\em reasoning probabilistically}. 
interpreting probabilities as epistemic entities suggests a Bayesian approach \cite{BayesEp}, 
where agents assess the likelihood of propositions based on a combination of prior assumptions and observations.

We provide a lightweight logic, the {\em aleatoric calculus}, for reasoning about systems of independent random variables,
and give an extension, the {\em modal aleatoric calculus} for reasoning about multi-agent systems of random variables.
We show that this is a true generalisation of modal logic and provide some initial proof theoretic results.
The modal aleatoric calculus allows agents to express strategies in games or theories of how other agents will act,
and we present a basic demonstration of this.

\section{Related Work}\label{sec:relatedwork}
There has been significant and long-standing interest in reasoning about probability and uncertainty,
to apply the precision of logical deduction in uncertain and random environments.
Hailperin's probability logic \cite{hailperin} and Nilsson's probabilistic logic \cite{nilsson} 
seek to generalise propositional, so the semantics of true and false are replaced by probability measures. 
These approaches in turn are generalised in fuzzy logics \cite{zadeh} where real numbers are used to model degrees of truth via T-norms.
In \cite{williamson} Williamson provide an inference system based on Bayesian epistemology.

These approaches lose the simplicity of Boolean logics, as deductive systems must deal with propositions that are not independent.
This limits their practicality as well defined semantics require the conditional probabilities of all atoms to be known.
However, these approaches have been successfully combined with logic programming \cite{hommerson} and machine learning \cite{tensorlog}.
Feldman and Harel \cite{feldman} and Kozen \cite{kozen} gave a probabilistic variation of propositional dynamic logic 
for reasoning about the correctness of programs with random variables. 
Importantly, this work  generalises a modal logic (PDL) as a many valued logic. 

More general foundational work on reasoning probabilistically was done by de Finetti \cite{deFinetti} 
who established an epistemic notion of probability based on what an agent would consider to be a rational wager (the {\em Dutch book} argument).
In \cite{Milne}, Milne incorporates these ideas into the logic of conditional events. 
Stalnaker \cite{stalnaker} has also considered conditional events and has presented conditional logic \cite{stalnakerThomason}. 
Here, conditional refers to the interpretation of one proposition being contingent on another, 
although this is not quantified nor assigned a probability.

The other approach to reasoning about uncertainty is to extend traditional Boolean and modal logics
with operators for reasoning about probabilities. 
Modal and epistemic logics have a long history for reasoning about uncertainty, going back to Hintikka's work on possible worlds \cite{hintikka}.
More recent work on dynamic epistemic logic \cite{plaza:1989,hvdetal.del:2007} has looked at how agents 
incorporate new information into thier belief structures.
There are explicit probabilistic extensions of these logics, that maintain the Boolean interpretation of formulae, 
but include probabilistic terms \cite{faginHalpern,halpern}. 
Probabilistic terms are converted into Boolean terms through arithmetic comparisons.
For example, ``It is more likely to snow than it is to rain'' is a Boolean statement, whereas the likelihood of snow is a probabilistic statement.

\section{Syntax and Semantics}\label{sect:semantics}


We take a many-valued approach here. Rather than presenting a logic that describes what is {\em true} about a probabilistic scenario,
we present the {\em Modal Aleatoric Calculus} (\mac) for determining what is likely. The different is subtle:
In probabilistic dynamic epistemic logic \cite{kooi} it is possible to express that the statement 
``Alice thinks X has probability 0.5'' is true;
whereas the calculus here simply has a term ``Alice's expectation of X'' which may have a value that is greater than 0.5.
We present a syntax for constructing complex terms in this calculus, and a semantics for assignment values to terms, 
given a particular interpretation or model. 

\subsection{Syntax}
The syntax is given for a set of random variables $X$, and a set of agents $N$.
We also include constants $\top$ and $\bot$. The syntax of the dynamic aleatoric calculus, \mac, is as follows:
$$ \alpha ::=\ x\ |\ \top\ |\ \bot\ |\ \ite{\alpha}{\alpha}{\alpha}\ |\ \cond{\alpha}{\alpha}_i$$
where $x\in X$ is a variable and $i\in N$ is a modality. We typically take an epistemic perspective, so the modality corresponds to an agent's beliefs. 
As usual, we let $v(\alpha)$ refer to the set of variables that appear in $\alpha$.
We refer to $\top$ as {\em always} and $\bot$ as {\em never}. 
The {\em if-then-else} operator $\ite{\alpha}{\beta}{\gamma}$ is read {\em if $\alpha$ then $\beta$ else $\gamma$} 
and uses the ternary conditional syntax of programming languages such as C.
The {\em conditional expectation} operator $\cond{\alpha}{\beta}_i$ is {\em modality $i$'s expectation of $\alpha$ given $\beta$} 
(the conditional probability $i$ assigns to $\alpha$ given $\beta$).

\subsection{Semantics}\label{semantics}
The modal aleatoric calculus is interpreted over {\em probability models} similar to the 
probability structures defined in \cite{halpern}, although they have random variables in place of propositional assignments.

\begin{definition}\label{def:PD}
  Given a set $S$, we use the notation $PD(S)$ to notate the set of {\em probability distributions} over $S$, where $\mu\in PD(S)$ implies:
  $\mu:S\longrightarrow[0,1]$; and either $\Sigma_{s\in S} \mu(s) = 1$, or $\Sigma_{s\in S}\mu(s) = 0$.
  In the latter case, we say $\mu$ is the {\em empty distribution}.
\end{definition}

\begin{definition}\label{def:pum}
  Given a set of variables $X$ and a set of modalities $N$, a {\em probability model} is specified by the tuple $P =(W, \pi, f)$, 
  where:
  \begin{itemize}
    \item $W$ is a set of possible worlds.
    \item $\pi:N\longrightarrow W\longrightarrow PD(W)$ assigns for each world $w\in W$ and each modality $i\in N$, 
      a probability distribution $\pi_i(w)$ over $W$.
      We will write $\pi_i(w,v)$ in place of $\pi(i)(w)(v)$.
    \item $f:W\longrightarrow X\longrightarrow[0,1]$ is a probability assignment so for each world $w$, for each variable $x$, 
      $f_w(x)$ is the probability of $x$ being true.
  \end{itemize}
  Given a model $P$ we identify the corresponding tuple as $(W^P, \pi^P, f^P)$. 
  A {\em pointed probability model}, $P_w = (W, \pi,f,w)$, specifies a world in the model as the point of evaluation.
\end{definition}

We note that we have not placed any restraints on the function $\pi$. 
If $\pi$ were to model agent belief we might expect all worlds in the probability distribution $\pi_i(w)$ to share the same probability distribution of worlds.
However, at this stage we have chosen to focus on the unconstrained case.

Given a pointed model $P_w$, the semantic interpretation of a \mac\ formula $\alpha$ is $P_w(\alpha)\in [0,1]$ 
which is the expectation of the formula being supported by a sampling of the model,
where the sampling is done with respect to the distributions specified by $\pi$ and $f$.
\begin{definition}\label{def:semantics}
  The {\em semantics of the modal aleatoric calculus} take a pointed probability model, $f_w$, 
  and a proposition defined in $\mac$, $\alpha$,  and calculate the {\em expectation of $\alpha$ holding at $P_w$}.
  Given an agent $i$, a world $w$ and a $\macLang$ formula $\alpha$, we define {\em $i$'s expectation of $\alpha$ at $w$} as
  $$E^i_w(\alpha) = \sum_{u\in W}\pi_i(w,u)\cdot P_u(\alpha).$$
  Then the semantics of $\mac$ are as follows:
  $$
  \begin{array}{rcl}
    P_w(\top) &=& 1  \qquad
    P_w(\bot)\ =\ 0  \qquad
    P_w(x)\ =\ f_w(x) \\ 
    P_w(\ite{\alpha}{\beta}{\gamma}) &=& P_w(\alpha)\cdot P_w(\beta)+(1-P_w(\alpha))\cdot P_w(\gamma)\\
    P_w(\cond{\alpha}{\beta}_i) &=& \frac{E^i_w(\alpha\land\beta)}{E^i_w(\beta)}\ 
       \textrm{if}\ E^i_w(\beta)>0\ \textrm{and}\ 1\ \textrm{otherwise}\\
  \end{array}
  $$
  We say two formulae, $\alpha$ and $\beta$, a {\em semantically equivalent} (written $\alpha\cong\beta$) if for all pointed probability models $P_w$ we have $P_w(\alpha) = P_w(\beta)$.
\end{definition}

The concept of {\em sampling} is intrinsic in the rational of these semantics. 
The word {\em aleatoric} has its origins in the Latin for dice-player ({\em aleator}), 
and the semantics are essentially aleatoric, in that they use dice (or sample probability distributions) for everything.
If we ask whether a variable is true at a world, the variable is sampled according to the probability distribution at that world.
Likewise, to interpret a modality the corresponding distribution of worlds is sampled, and the formula is evaluated at the selected world.
However, we are not interested in the result of any one single sampling activity, but in the {\em expectation} derived from the sampling activity.

This gives us an interesting combination approaches for understanding probability. 
Aleatoric approaches appeal to frequentist interpretations of probability, 
where likelihoods are fixed and assume arbitrarily large sample sizes. 
This contrasts the Bayesian approach where probability is intrinsically epistemic, 
where we consider what likelihood an agent would assign to an event, given the evidence they have observed.
Our approach can be seen as an aleatoric implementation of a Bayesian system. By this we mean that:
\begin{itemize}
  \item Random variables are {\em Aleatoric}, always sampled from a fixed distribution.
  \item Modalities are {\em Bayesian}, conditioned on a set of possible worlds.
  \item Possible worlds are {\em Aleatoric}, sampled from a probability distribution.
\end{itemize}

We can imagine many different scenarios this way, and in some scenarios $\alpha$ is true, and in some scenarios, $\alpha$ is false.
The aleatoric calculus determines the likelihood of a formula being true in this sampling process.
What is convenient about this process is that every sampling event is independent of every other sampling event. 
Consequently, two subformulae can be evaluated independently, allowing a simple mathematical interpretation for the formulae.
We also restrict agents reasoning capability solely to this sampling process. 
We cannot compare probabilities directly to, say, reason that $\alpha$ is twice as likely as $\beta$.
However,  we express that an agent considers a formula is impossible (i.e. $\cond{\bot}{\alpha}_i$ is true).
and we can repeatedly sample formulae to get better approximations 
(e.g. $(\beta\rightarrow\alpha)^{n/2n}$ tends to 0 with $n$ if and only if $\alpha$ is less likely than $\beta$).

With the concept of agents sampling their mental model of the universe in mind, 
we see that sampling $\top$ {\em always} returns true, and sampling $\bot$ {\em never} returns true.
Sampling a random variable will return true in line with the probability given by $f$. 
The variable will be resampled every time it appears in formula, 
so the expectation for $x\land x$ is typically {\em not} the same as the expectation for $x$.
Therefore a random variable with probability greater than 0 and less than 1 should only model an independent stochastic event, 
like flipping a coin.
However, if a game involves flipping a coin and leaving it hidden whilst players bet on its orientation, 
we should not consider each reference to the coin's orientation as an independent event. 
Rather, we should consider two possible worlds: one where the coin is heads; and one where the coin is tails.

The if-then-else operator, $\ite{\alpha}{\beta}{\gamma}$, can be imagined as a sampling protocol. 
We first sample $\alpha$, and if $\alpha$ is true,
we proceed to sample $\beta$ and otherwise we sample $\gamma$. 
We imagine an evaluator interpreting the formula by flipping a coin:
if it lands heads, we evaluate $\beta$; if it lands tails, we evaluate $\gamma$.
This corresponds to the additive property of Bayesian epistemology: 
{\em if $A$ and $B$ are disjoint events, then $P(A\ \textrm{or}\ $B$) = P(A)+P(B)$}\cite{BayesEp}.
Here the two disjoint events are $\alpha$ {\em and} $\beta$ and $\lnot\alpha${\em and} $\gamma$, 
but disjointedness is only guaranteed if $\alpha$ and $\lnot\alpha$ are evaluated from the same sampling. 

The conditional expectation operator $\cond{\alpha}{\beta}_i$ expresses modality $i$'s expectation of $\alpha$ {\em marginalised} by the expectation of $\beta$.
This is, as in the Kolmogorov definition of conditional probability, 
$i$'s expectation of $\alpha\land\beta$ divided by $i$'s expectation of $\beta$.
The intuition for these semantics corresponds to a sampling protocol. 
The modality $i$ samples a world from the probability distribution and samples $\beta$ at that world. 
If $\beta$ is true, then $i$ samples $\alpha$ at that world and returns the result. 
Otherwise agent $i$ resamples a world from their probability distribution, and repeats the process.
In the case that $\beta$ is never true, we assign $\cond{\alpha}{\beta}_i$ probability 1, as being vacuously true.


\paragraph{Abbreviations:}
Some abbreviations we can define in $\macLang$ are as follows:
$$
\begin{array}{c}
  \alpha\land\beta = \ite{\alpha}{\beta}{\bot} \qquad  
  \alpha\lor\beta = \ite{\alpha}{\top}{\beta} \qquad
  \alpha\imp\beta  = \ite{\alpha}{\beta}{\top} \qquad
  \lnot\alpha  =  \ite{\alpha}{\bot}{\top}\\
  \expect_i\alpha  =  \cond{\alpha}{\top}_i \qquad 
  \believes_i \alpha = \cond{\bot}{\lnot\alpha}_i\\
  \alpha^{\frac{0}{b}} = \top \qquad 
  \alpha^{\frac{a}{b}} = \bot\ \textrm{if}\ b<a\qquad 
  \alpha^{\frac{a}{b}} =\ite{\alpha}{\alpha^{\frac{a-1}{b-1}}}{\alpha^{\frac{a}{b-1}}}\ \textrm{if}\ a\leq b
\end{array}
$$
where $a$ and $b$ are natural numbers. 
We will show later that under certain circumstances these operators do correspond with their Boolean counterparts.
However, this is not true in the general case. 
The formula $\alpha\land\beta$ does not interpret directly as $\alpha$ is true and $\beta$ is true.
Rather it is the likelihood of $\alpha$ being sampled as true, followed by $\beta$ being sampled as true.
For this reason $\alpha\land\alpha$ is {\em not} the same as $\alpha$.
Similarly $\alpha\lor\beta$ is the likelihood of $\alpha$ being sampled as true, 
or in the instance that it was not true, that $\beta$ was sampled as true.

The modality $\expect_i \alpha$ is agent $i$'s expectation of $\alpha$ being true, 
which is just $\alpha$ conditioned on the uniformly true $\top$.
The operator $\believes_i\alpha$ corresponds to the necessity operator of standard modal logic, 
and uses a property of the conditional operator:
it evaluates $\cond{\alpha}{\beta}_i$ as vacuously true if and only if there is no expectation that $\beta$ can ever be true.
Therefore, $\cond{\bot}{\lnot\alpha}_i$ can only be true if modality $i$ always expects $\lnot\alpha$ to be false, 
and thus for the modality $i$, $\alpha$ is necessarily true.
The formula $\alpha^{\frac{a}{b}}$ allows us to explicitly represent degrees of belief in the language.
It is interpreted as {\em $\alpha$ is true at least $a$ times out of $b$}.
Note that this is not a statement saying what the frequency of $\alpha$ is. Rather it describes the event of $\alpha$ being true $a$ times out of $b$.
Therefore, if $\alpha$ was unlikely (say true 5\% of the time) then $\alpha^{\frac{9}{9}}$ describes an incredibly unlikely event.

\subsection{Example}
We will give simple example of reasoning in \mac.
Suppose we have an aleator (dice player), considering the outcome of a role of a die. 
While the dice is fair, our aleator does not know whether it is a four sided die or a six sided die.
We consider a single proposition: $p_1$ if the result of throw of the die is 1.
The aleator considers two worlds equally possible: $w_4$ where the die has four sides, and $w_6$ where the die has $6$ sides.
The probability model $P = (W, \pi, f)$ is depicted in Figure~\ref{dice}:
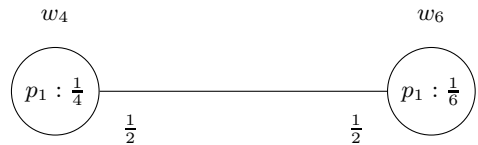
\begin{figure}
    \begin{center}
  \begin{tikzpicture}
      \draw (0,0) node[circle,draw](w4) {$p_1:\frac{1}{4}$};
      \draw (0,1) node {$w_4$};
      \draw (5,0) node[circle,draw](w6) {$p_1:\frac{1}{6}$};
      \draw (5,1) node {$w_6$};
      \draw (w4)--(w6);
      \draw (1,-0.5) node {$\frac{1}{2}$};
      \draw (4,-0.5) node {$\frac{1}{2}$};
  \end{tikzpicture}
    \end{center}
    \caption{A probability model for an aleator who does not know whether the die is four sided ($w_4$) or six sided ($w_6$).}\label{dice} 
\end{figure}
  We can formulate properties such as ``at least one of the next two throws will be a 1'': $p_1^{\frac{1}{2}} = \ite{p_1}{\top}{p_1}$. 
  We can calculate $P_{w_4}(p_1^{\frac{1}{2}}) = \frac{7}{16}$, while $P_{w_6}(p_1^{\frac{1}{2}}) = \frac{11}{36}$.
  Now if we asked our aleator what are the odds of rolling a second 1, given the first roll was 1, we would evaluate the formula
  $\cond{p_1}{p_1}_a$ (where $a$ is our aleator), and in either world this evaluates to $\frac{5}{24}$.
  Note that this involves some speculation from the aleator.

\section{Axioms for the modal aleatoric calculus}\label{axioms}

Having seen the potential for representing stochastic games, we will now look at some reasoning techniques.
First we will consider some axioms to derive constraints on the expectations of propositions, 
as an analogue of a Hilbert-style proof system for modal logic.
In the following section we will briefly analyse the model checking problem, as a direct application of the semantic definitions. 

Our approach here is to seek a derivation system that can generate equalities that are always valid in \mac. 
For example, $\alpha\land\beta\simeq\beta\land\alpha$ will be satisfied by every world of every model.
We use the relation $\simeq$ to indicate that two formulae are equivalent in the calculus, 
and the operator $\cong$ to indicate the expectation assigned to each formula will be equal in all probability models.
We show that the calculus is sound, and sketch a proof of completeness in the most basic case.  

\subsection{The aleatoric calculus}
The {\em aleatoric calculus}, \ac, is the language of $\top,\ \bot,\ x$ and $\ite{\alpha}{\beta}{\gamma}$, where $x\in X$. 
The interpretation of this fragment only depends on a single world and it is the 
analogue of propositional logic in the non-probabilistic setting. 
The axioms of the calculus are:
$$\begin{array}{lrcllrcl}
  {\bf id}& x &\simeq& x  &\qquad
  {\bf vacuous}&\ite{x}{\top}{\bot} &\simeq& x\\
  {\bf ignore}& \ite{x}{y}{y} &\simeq& y  &\qquad
  {\bf tree}&\ite{\ite{x}{y}{z}}{p}{q} &\simeq& \ite{x}{\ite{y}{p}{q}}{\ite{z}{p}{q}}\\
  {\bf always}& \ite{\top}{x}{y} &\simeq& x  &\qquad
  {\bf swap}&\ite{x}{\ite{y}{p}{q}}{\ite{y}{r}{s}} &\simeq& \ite{y}{\ite{x}{p}{r}}{\ite{x}{q}{s}}\\
  {\bf never}& \ite{\bot}{x}{y} &\simeq& y  &&&&
\end{array}
$$
We also have the rules of transitivity, symmetry and substitution for $\simeq$:
$$
\begin{array}{l}
  \mathbf{Trans}:\qquad \textrm{If }\alpha\simeq\beta\textrm{ and }\beta\simeq\gamma\textrm{ then }\alpha\simeq\gamma\\ 
  \mathbf{Sym}:\qquad \textrm{If } \alpha\simeq \beta\textrm{ then } \beta\simeq\alpha\\
  \mathbf{Subs}:\qquad \textrm{If } \alpha \simeq \beta\textrm{ and } \gamma\simeq\delta
  \textrm{ then } \alpha[x\backslash\gamma] \simeq \beta[x\backslash\delta]
\end{array}
$$
where $\alpha[x\backslash\gamma]$ is $\alpha$ with every occurrence of the variable $x$ replaced by $\gamma$.
We let this system of axioms and rules be referred to as \acAx.

As an example of reasoning in this system, we will show that the commutativity of $\land$ holds:
$$
\begin{array}{rcll}
    \ite{x}{y}{\bot} & \simeq & \ite{x}{\ite{y}{\top}{\bot}}{\ite{y}{\bot}{\bot}} & \textbf{vacuous}, \textbf{ignore} \\
    & \simeq & \ite{y}{\ite{x}{\top}{\bot}}{\ite{x}{\bot}{\bot}} & \textbf{swap} \\
    & \simeq & \ite{y}{x}{\bot} & \textbf{vacuous}, \textbf{ignore} \\
\end{array}
$$

The axiom system \acAx\ is sound.
The majority of these axioms are simple to derive from Definition~\ref{def:semantics}, 
and all proofs essentially show that the semantic interpretation of the left and right side of the equation are equal.
The rules {\bf Trans} and {\bf Sym} come naturally with equality, 
and the rule {\bf Subs} follows because at any world, all formulae are probabilistically independent.

We present arguments for the soundness of the less obvious \textbf{tree} and \textbf{swap} below. 
As the aleatoric calculus is only interpreted with respect to a single world, 
we will drop the subscript from the evaluation function $P_w$. 
We will also use the abbreviation $\lnot\alpha$ for $\ite{\alpha}{\bot}{\top}$ noting $P(\lnot\alpha) = 1-P(\alpha)$.

\begin{lemma} \label{lem:not_distributes}
      $P(\lnot{\ite{x}{\alpha}{\beta}}) = P(\ite{x}{\lnot{\alpha}}{\lnot{\beta}})$
\end{lemma}

\begin{proof}
  $$
  \begin{array}{rl}
      P(\lnot{\ite{x}{y}{z}}) &= 1 - P(\ite{x}{y}{z}) \\
          &= 1 - P(x) P(y) - P(\lnot{x}) P(z) \\
          &= 1 - P(x) (1 - P(\lnot{y})) - P(\lnot{x}) (1 - P(\lnot{z})) \\
          &= 1 - P(x) + P(x) P(\lnot{y}) - P(\lnot{x}) + P(\lnot{x}) P(\lnot{z})) \\
          &= P(x) P(\lnot{y}) + P(\lnot{x}) P(\lnot{z})) \\
          &= P(\ite{x}{\lnot{y}}{\lnot{z}})
    \end{array}
    $$
\end{proof}

\begin{lemma}\label{lem:treeSound}
      The \textbf{tree} axiom is sound with respect to $\simeq$.
\end{lemma}

\begin{proof}
    By Definition~\ref{def:semantics} and Lemma~\ref{lem:not_distributes}, we can show that
    $$\begin{array}{rl}
      &P(\ite{\ite{x}{y}{z}}{p}{q})\\ 
      &= P(\ite{x}{y}{z}) P(p) + P(\lnot{\ite{x}{y}{z}}) P(q) \\
      &= P(\ite{x}{y}{z}) P(p) + P(\ite{x}{\lnot{y}}{\lnot{z}}) P(q) \\
      &= P(x) P(y) P(p) + P(\lnot{x}) P(z) P(p) + 
         P(x) P(\lnot{y}) P(q) + P(\lnot{x}) P(\lnot{z}) P(q) \\
      &= P(x) (P(y) P(p) + P(\lnot{y}) P(q)) + 
         P(\lnot{x}) (P(z) P(p) + P(\lnot{z}) P(q)) \\
      &= P(x) P(\ite{y}{p}{q}) + P(\lnot{x}) P(\ite{z}{p}{q}) \\
      &= P(\ite{x}{\ite{y}{p}{q}}{\ite{z}{p}{q}})
    \end{array}
    $$
    for any $P : X \rightarrow [0, 1]$. 
    By Definition~\ref{def:semantics} this implies
   $$
   \ite{\ite{x}{y}{z}}{p}{q} \simeq \ite{x}{\ite{y}{p}{q}}{\ite{z}{p}{q}}
   $$
   which is the \textbf{tree} axiom.
\end{proof}

\begin{lemma}\label{lem:swapSound}
      The \textbf{swap} axiom is sound with respect to $\simeq$.
\end{lemma}

\begin{proof}
  By Definition~\ref{def:semantics}, we can show that
  $$
  \begin{array}{rl}
    &P(\ite{x}{\ite{y}{z}{p}}{\ite{y}{r}{z}})\\ 
    &= P(x)P(\ite{y}{z}{p}) + P(\lnot{x})P(\ite{y}{r}{z}) \\
    &= P(x)P(y)P(z) + P(x)P(\lnot{y})P(p) + P(\lnot{x})P(y)P(r) + 
       P(\lnot{x})P(\lnot{y})P(z) \\
    &= P(y)(P(x)P(z) + P(\lnot{x})P(r)) + P(\lnot{y})(P(x)P(p) + P(\lnot{x})P(z)) \\
    &= P(y)P(\ite{x}{z}{r}) + P(\lnot{y})P(\ite{x}{p}{z}) \\
    &= P(\ite{y}{\ite{x}{z}{r}}{\ite{x}{p}{z}})
  \end{array}
  $$
  for any $P : X \rightarrow [0, 1]$. By \ref{def:semantics} this implies
  $$
  \ite{x}{\ite{y}{z}{p}}{\ite{y}{r}{z}} \simeq \ite{y}{\ite{x}{z}{r}}{\ite{x}{p}{z}}
  $$
  which is the \textbf{swap} axiom.
\end{proof}

\begin{lemma}\label{bacAx-sound}
  The axiom system \acAx\ is sound for \ac.
\end{lemma}
\begin{proof}
  This follows from Lemmas~\ref{lem:treeSound}~and~\ref{lem:swapSound} and Definition~\ref{semantics}.
  Also, from the semantics we can see that the interpretation of subformulae are independent of one another, 
  so the substitution rule holds, and the remaining rules follow directly from the definition of $\simeq$.
\end{proof}

To show that \acAx\ complete for the aleatoric calculus, 
we aim to show that any aleatoric calculus formula that are semantically equivalent
can be transformed into a common form. 
As the axioms of \acAx\ are equivalences this is sufficient to show that the formulae are provably equivalent.

A \emph{tree form} Aleatoric Calculus formula is either atomic, or it has an atomic random variable condition and both its left and right subformulae are in tree form.

\begin{definition} \label{def:tree_form}
    The set of all \emph{tree form} Aleatoric Calculus formulae $T \subset \Phi$ is generated by the following grammar:
    $$ \varphi ::= \top \mid \bot \mid \ite{x}{\varphi}{\varphi} $$
\end{definition}

\begin{lemma} \label{lem:tf}
    For any Aleatoric Calculus formula there exists an equivalent (by $\simeq$) tree form formula.
\end{lemma}

\begin{proof}
    Boolean atoms ($\top$ and $\bot$) are in tree form by definition.
    By \textbf{vacuous}, any atomic random variable $x$ has an equivalent tree form $\ite{x}{\top}{\bot}$.
    Hence any atomic formula has an equivalent tree form.

    For non-atomic formula, we can show any compound aleatoric calculus formula $\ite{\mu}{\gamma}{\delta}$ 
    is $\simeq$-equivalent to some formula with an atomic condition (i.e. $\ite{x}{\nu_1}{\nu_2}$).
    Consider some compound formula $\ite{\mu}{\gamma}{\delta}$.
    Let us assume there exists some $\ite{x}{\alpha}{\beta} \simeq \mu$ where $x$ is atomic.
    \begin{flalign*}
        && \ite{\mu}{\gamma}{\delta} &\simeq \ite{\ite{x}{\alpha}{\beta}}{\gamma}{\delta} && \textbf{Subs} \\
        && &\simeq \ite{x}{\ite{\alpha}{\gamma}{\delta}}{\ite{\beta}{\gamma}{\delta}} && \textbf{tree}
    \end{flalign*}
    It follows by induction that any compound formulae has an equivalent formula with an atomic condition.
    Furthermore we may assume that $x$ is not $\top$ or $\bot$ since
    \begin{flalign*}
        && \ite{\top}{\alpha}{\beta} &\simeq \alpha && \textbf{always}, \textbf{Sym} \\
        && \ite{\bot}{\beta}{\alpha} &\simeq \alpha && \textbf{never}, \textbf{Sym}
    \end{flalign*}
    Therefore, if the left and right subformulae of a such a compound formula have equivalent tree forms, 
    by \textbf{Subs} we can susbtitute them into the original formula to produce an equivalent tree form.
    It follows by induction that for any Aleatroic Calculus formula there exists an equivalent tree form.
\end{proof}

\begin{figure}
    \centering
    \begin{forest}
        [$x$ [$y$ [$\alpha$] [$\beta$]] [$\gamma$]]
    \end{forest}
    \caption{The graphical representation of the formula $\ite{x}{\ite{y}{\alpha}{\beta}}{\gamma}$.}
    \label{fig:tree_example}
\end{figure}
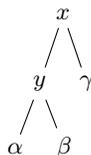

For ease of interpretation, we may represent tree form aleatroic calculus expressions graphically as in Figure~\ref{fig:tree_example}.

\begin{definition}
    A \emph{path} in a tree form aleatoric calculus formula is a sequence of tokens from the set $\set{x, \nol x \mid x \in X}$ 
    corresponding to the outcomes of random trials involved in reaching a terminal node in the tree.
    We define the functions $\Top$ and $\Bot$ to be the set of paths that terminate in a $\top$ or a $\bot$ respectively:
    $$
    \Top(\top) = \Bot(\bot) = \{\seq{}\}, \qquad
    \Top(\bot) = \Bot(\top) = \emptyset,
    $$
    $$
    \Top(\ite{x}{\alpha}{\beta}) = \set{\seq{x} ^\frown a \mid a \in \Top(\alpha)} \cup \set{\seq{\nol x} ^\frown b \mid b \in \Top(\beta)}
    $$
    $$
    \Bot(\ite{x}{\alpha}{\beta}) = \set{\seq{x} ^\frown a \mid a \in \Bot(\alpha)} \cup \set{\seq{\nol x} ^\frown b \mid b \in \Bot(\beta)}
    $$
    where $^\frown$ is the sequence concatenation operator:
    $$
    \seq{a_1,\dots,a_n} ^\frown \seq{b_1,\dots,b_n} = \seq{a_1,\dots,a_n, b_1, \dots, b_n}
    $$
    We say two paths are {\em multi-set equivalent} if the multiset of tokens that appear in each path are equivalent, 
    and define $\Path(\phi) =\Top(\phi)\cup\Bot(\phi)$ to be the set of all paths through a formula.
\end{definition}

\begin{lemma} \label{lem:path_product}
    For any tree form aleatoric calculus formula $\alpha$:
    $$
    P(\alpha) = \sum_{t \in \Top(\alpha)}{\prod_{x \in t}{P(x)}}
    $$
    where $P(\overline{x}) = (1 - P(x))$.
\end{lemma}

\begin{proof}
  This follows immediately from the Definition~\ref{def:semantics}.
\end{proof}

\begin{figure}
    \centering
    \begin{forest}
        [$x$
            [$x$
                [$\seq{x, x}$]
                [$y$
                    [$\seq{x, \nol x, y}$]
                    [$\seq{x, \nol x, \nol y}$]
                ]
            ]
            [$x$
                [$y$
                    [$\seq{\nol x, x, y}$]
                    [$\seq{\nol x, x, \nol y}$]
                ]
                [$\seq{\nol x, \nol x}$]
            ]
        ]
    \end{forest}
    \caption{An example of paths for a tree form, corresponding to 
      $\phi = \ite{x}{\ite{x}{\bot}{y}}{\ite{x}{y}{\bot}}$. 
      The set $\Top(\phi)$ is $\{\seq{x,\nol x, y},\seq{\nol x,x,y}\}$.}
    \label{fig:path_example}
\end{figure}
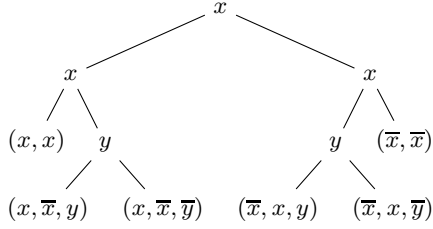

\begin{lemma} \label{lem:path_perm}
  Suppose that $\phi$ is a formula in tree form such that $a = (a_0,\hdots,a_n)\in\Top(\phi)$ (resp. $\Bot(\phi)$).
  Then, for any $i<n$ there is some formula $\phi_a^i$ such that:
  \begin{enumerate}
    \item $\phi\simeq\phi_a^i$
    \item $(a_0,\hdots,a_{i-1},a_{i+1},a_i,a_{i+2},\hdots,a_n)\in \Top(\phi_a^i)$ (resp. $\Bot(\phi)$)
    \item $\phi$ and $\phi_a^i$ agree on all paths that do not have the prefix $(a_0,\hdots, a_{i-1}$. 
      That is, for all $b\in\Path(\phi)\cup\Path(\phi_a^i)$, where for some $j<i$, $b_j \neq a_j$, 
      we have $b\in\Top(\phi)$ if and only if $b\in\Top(\phi_a^i)$ and $b\in\Bot(\phi)$ if and only if $b\in\Bot(\phi_a^i)$. 
  \end{enumerate}
\end{lemma}

\begin{proof}
    This is trivially true for an empty path or a path of unit length, as there is no way to reorder either such path.
    Let us now assume this statement holds for a path of length $n$.
    Consider a path $a$ of length $n+1$, $n > 0$, and suppose $i\leq n$ is given.
    Suppose $a_i\in\{x,\nol x\}$, and $a_{i+1}\in \{x,\nol y\}$. 
    If $a_i = x$ and $a_{i+1} = y$ then $\phi$ has a subformula $\ite{x}{\ite{y}{\alpha}{\beta}}{\gamma}$
    with the path leading to the formula matching $(a_0,\hdots,i)$. 
    Similar subformulae can be found for the other possibilities of $a_i$ and $a_{i+1}$.
    We can consider this path as the concatenation of a unit path (the head) and a path of length $n$ (the tail):
    $$
    a = \seq{a_0} ^\frown \seq{a_1, \dots, a_n}
    $$
    Let $a^\prime = \seq{a_0^\prime, a_1^\prime, \dots, a_n^\prime}$ be some arbitrary target permutation of $a$.
    If $a_0^\prime = a_0$, we can simply reorder $\seq{a_1, \dots, a_n}$ as desired.
    Otherwise, it follows that $a_0^\prime$ must appear somewhere in the tail of $a$.
    We may therefore permute the tail of $a$ such that $a_0^\prime$ is at its head:
    $$
    \seq{a_0} ^\frown \seq{a_0^\prime, \dots}
    $$
    Let us assume that $a_0 \in \set{x, \nol x}$, $a_0^\prime \in \set{y, \nol y}$, and that the tail of $\seq{a_0^\prime, \dots}$ is a path through a subtree $\alpha$.
    This gives four possible tree structures, all of which may be rearranged using \textbf{ignore} and \textbf{swap}:
    $$
    \ite{x}{\ite{y}{\alpha}{\beta}}{\gamma} \simeq \ite{x}{\ite{y}{\alpha}{\beta}}{\ite{y}{\gamma}{\gamma}} \simeq \ite{y}{\ite{x}{\alpha}{\gamma}}{\ite{x}{\beta}{\gamma}}
    $$
    $$
    \ite{x}{\ite{y}{\beta}{\alpha}}{\gamma} \simeq \ite{x}{\ite{y}{\beta}{\alpha}}{\ite{y}{\gamma}{\gamma}} \simeq \ite{y}{\ite{x}{\beta}{\gamma}}{\ite{x}{\alpha}{\gamma}}
    $$
    $$
    \ite{x}{\gamma}{\ite{y}{\alpha}{\beta}} \simeq \ite{x}{\ite{y}{\gamma}{\gamma}}{\ite{y}{\alpha}{\beta}} \simeq \ite{y}{\ite{x}{\gamma}{\alpha}}{\ite{x}{\gamma}{\beta}}
    $$
    $$
    \ite{x}{\gamma}{\ite{y}{\beta}{\alpha}} \simeq \ite{x}{\ite{y}{\gamma}{\gamma}}{\ite{y}{\beta}{\alpha}} \simeq \ite{y}{\ite{x}{\gamma}{\beta}}{\ite{x}{\gamma}{\alpha}}
    $$
    Combined with the substitution rule we define $\phi_a^i$ as the reulst of performing these transformations with in $\phi$.
    Then $\phi_a^i\simeq\phi$, and as $\alpha$ remains the subformula after $a_{i+1}$ and $a_i$, the suffix of the path remains the same.
    Finally, as the substitution will only affect the subtree where the substitution occurs, the paths without the prefix $(a_0,\hdots,a_i)$ remain unchanged.
\end{proof}
Therefore, we are able repeatdly apply this construcion to permute the elements of paths within a formula.

\begin{lemma} \label{lem:path_swap}
  Given a pair of multi-set equivalent paths $a$ and $b$ in a tree form aleatoric calculus formula, $\phi$, such that $a\in \Top(\phi)$ and $b\in\Bot(\phi)$, 
  we can find a formula $\phi_a^b\simeq\phi$ where 
  \begin{enumerate}
    \item $a\in \Bot(\phi_a^b)$ and $b\in \Top(\phi_a^b)$
    \item $\Top(\phi_a^b)-\{b\} = \Top(\phi)-\{a\}$, 
    \item $\Bot(\phi_a^b)-\{a\} = \Bot(\phi)-\{b\}$.
  \end{enumerate}
\end{lemma}

\begin{proof}
    Suppose $a$ and $b$ have length $2$.
    There are exactly two permutations for a path of length two, the original, and its reverse.
    If the paths are of the form $\seq{x, \nol x}$ and $\seq{\nol x, x}$ respectively, then it follows that the tree must be of the form $\ite{x}{\ite{x}{\alpha}{\beta}}{\ite{x}{\gamma}{\delta}}$, and so the paths can be swapped using \textbf{swap}.
    Let us now assume that this statement holds for all paths of length at most $n > 1$.
    Consider a pair of multiset-equivalent paths of length $n+1$ in $\ite{x}{\mu}{\nu}$.
    If both paths start with $x$, then we can remove this $x$ from both and instead have two multiset-equivalent paths of length $n$ in $\mu$, which by our earlier assumption we are able to swap.
    Similarly, if both paths start with $\nol x$ we can remove it from both and recurse into $\nu$.
    Otherwise, one path must start with $x$, and pass through $\mu$, while the other starts with $\nol x$, and passes through $\nu$.
    Let us refer to these paths as $l = \seq{x, t, \dots}$, $t \in \set{y, \nol y}$ and $r = \seq{\nol x, r_1, \dots}$, respectively.
    Since $l$ and $r$ are multiset-equivalent, it follows that $r$ must contain at least one $t$ term.
    By \cref{lem:path_perm}, we are able to rearrange $\nu$ to permute the subpath of $r$ it contains in order to bring a $t$ term to the head of the subpath.
    Note that this rearrangement is reversible.
    This gives a permutation $\seq(\nol x, t, \dots)$ of $r$.
    Since both paths now have the same second element, $t \in \set{y, \nol y}$, it follows that both subtrees must have $y$ at their root.
    We are therefore able to apply \textbf{swap}:
    $$
    \ite{x}{\ite{y}{\alpha}{\beta}}{\ite{y}{\gamma}{\delta}} \simeq \ite{y}{\ite{x}{\alpha}{\gamma}}{\ite{x}{\beta}{\delta}} 
    $$
    This in turn swaps the first two elements of $l$ and $r$ to give $\seq{t, x, \dots}$ and $\seq{t, \nol x, \dots}$.
    Since both paths now start with the same element, $t$, and are therefore in the same subtree, we can remove this $t$ term and instead have two multiset-equivalent paths of length $n$ in either the left or right subtree, which we are able to swap as per our earlier assumption.
    We are then able to replace $t$ in both paths and reverse our previous rearrangements.
    Since only the leaves were swapped, reversing the process gives us the original tree except that the two leaves corresponding to the given paths have been swapped.
    We have therefore shown that if this statement holds true for a pair of paths of length $n > 1$, it must also hold for a pair of paths of length $n+1$.
\end{proof}

\begin{lemma} \label{lem:tf_same_structure}
    For any pair of tree form aleatoric calculus formulae, $\phi$ and $\psi$, there exists a pair of tree forms $\phi'\simeq\phi$ and $\psi'\simeq\psi$,
    such that $\Path(\phi') = \Path(\psi')$.
\end{lemma}

\begin{proof}
    Let us consider a pair of tree forms $L$ and $R$ of heights $h_L$ and $h_R$, respectively.
    In the case that $m = n = 0$, both trees are already atomic leaves, and so the statement holds.
    Next consider the case that $L = \ite{x}{\alpha}{\beta}$ and $h_R = 0$.
    We can construct $\ite{x}{R}{R} \simeq R$, and therefore require the statement holds for the pairs $\alpha$ and $R$ and $\beta$ and $R$.
    Since $\alpha$ and $\beta$ must be shorter than $L$, it follows by induction that the statement holds for all $h_L >= 0, h_R = 0$.
    A similar process works to show the statement holds for all $h_L = 0, h_R >= 0$.
    Next consider the case that $L = \ite{x}{\alpha}{\beta}$ and $R = \ite{y}{\gamma}{\delta}$.
    We may construct $L \simeq \ite{y}{L}{L}$ and require that the statement holds for the pairs $L$ and $\gamma$ and $L$ and $\delta$.
    Since $\gamma$ and $\delta$ must be shorter than $R$, and we have shown that the statement holds for all $h_L >= 0, h_R = 0$, 
    it follows by induction that the statement holds for all $h_L >= 0, h_R >= 0$, and hence for all pairs of tree form aleatoric calculus formulae.
\end{proof}

\begin{theorem} \label{thm:completeness}
    For any pair of semantically equivalent aleatoric calculus formulae $\phi$ and $\psi$, we can show $\phi\simeq\psi$.
\end{theorem}

\begin{proof}
    By Lemma~\ref{lem:tf} it is possible to convert both formulae in question, $\phi$ and $\psi$, 
    into tree form, respectively $\phi^\tau$ and $\psi^\tau$.
    By Lemma~\ref{lem:tf_same_structure} it is then possible to convert $\phi^\tau$ and $\psi^\tau$ to a pair of equivalent formulae, 
        respectively $\Phi$ and $\Psi$, with the same structure (so $\Top(\Phi)\cup\Bot(\Phi) = \Top(\Psi)\cup\Bot(\Psi)$), 
        but possibly different leaves (so $\Top(\Phi)$ is possibly not the same as $\Top(\Psi))$.
    By Lemma~\ref{lem:path_swap} it is possible to swap any multiset-equivalent paths between $\Top(\Phi)$ and $\Bot(\Phi)$.
    By Lemma~\ref{lem:path_product} two formula, $\Phi$ and $\Psi$, with the same structure are semantically equivalent if and only if 
    there is a one-to-one correspondence between paths of $\Phi$ and $\Psi$ such that 
    corresponding paths $a$ and $b$ are multi-set equivalent, 
    and $a\in\Top(\Phi)$ if and only if $b\in\Top(\Psi)$. 
    Therefore, if and only if the two formulae are equivalent we are able to define $\Phi'$ by 
    swapping paths between $\Top(\Phi)$ and $\Bot(\Phi)$ such that $\Phi' = \Psi$.
    As all steps are performed using the axioms and are reversible, this is sufficient to show $\phi\simeq\psi$.
\end{proof}

\subsection{The Modal Aleatoric Calculus}
The modal aleatoric calculus includes the propositional fragment, as well as the conditional expectation operator $\cond{\alpha}{\beta}_i$ 
that depends on the modality $i$'s probability distribution over the set of worlds. 

The axioms we have for the conditional expectation operator are as follows:
$$\begin{array}{rrcl}
  {\bf A0:}&\qquad \cond{\ite{x}{y}{z}}{c}_i &\simeq& \ite{\cond{x}{c}_i}{\cond{y}{\ite{x}{c}{\bot}}_i}{\cond{z}{\ite{x}{\bot}{c}}_i}.\\
  {\bf A1:}&\qquad \cond{\bot}{x}_i\land\cond{x}{y}_i &\simeq& \cond{\bot}{x\lor y}_i\\
  {\bf A2:}&\qquad \cond{\bot}{x}_i &\simeq& \ite{\cond{\bot}{x}_i}{\cond{\bot}{x}_i}{\lnot\cond{\bot}{x}_i}\\
  {\bf A3:}&\cond{\top}{x}_i &\simeq& \top\\
  {\bf A4:}&\cond{x}{\bot}_i &\simeq& \top
\end{array}$$
We let the axiom system \macAx\ be the axiom system \acAx\ along with the axioms {\bf A0}-{\bf A5}.

We note that the conditional expectation operator $\cond{x}{y}_i$ is its own dual, 
but only in the case that agent $i$ does not consider $x$ and $y$ to be mutually exclusive: 
$\cond{\lnot x}{y}_i\simeq \cond{x}{y}_i\rightarrow \believes_i(\lnot(x\land y)).$
We can see this in the following derivation:
$$\begin{array}{rclr}
  \cond{\lnot x}{y}_i&\simeq&{\cond{\ite{x}{\bot}{\top}}{y}_i} &\textrm{abb.}\\
  &\simeq& \ite{\cond{x}{y}_i}{\cond{\bot}{x\land y}_i}{\cond{\top}{\ite{x}{\bot}{\bot}}_i}&\textrm{A0}\\
  &\simeq& \cond{x}{y}_i\rightarrow\believes_i\lnot(x\land y)&\textrm{abb.}
\end{array}$$

The main axiom in $\macAx$ is the axiom {\bf A0} which is a rough analogue of the {\bf K} axiom in modal logic. 
We note that in this axiom: 
$$\cond{\ite{x}{y}{z}}{c}_i\simeq\ite{\cond{x}{c}_i}{\cond{y}{\ite{x}{c}{\bot}}_i}{\cond{z}{\ite{x}{\bot}{c}}_i}$$
if we substitute $\top$ for $y$ and $\bot$ for $w$, we have:
$$\expect_i x\land\cond{y}{x}_i\simeq\expect_i(x\land y)$$
whenever agent $i$ considers $x$ possible (so that $\cond{\bot}{\lnot x}_i\simeq\bot$). 
In that case we can ``divide'' both sides of the semantic equality by $P_w(\expect_i x)$ 
which gives the Kolmogorov definition of conditional probability: 
$$P_w(\cond{y}{x}_i) = \frac{P_w(\expect_i(x\land y))}{P_w(\expect_i x)}.$$

Axioms {\bf A1} and {\bf A2} deal with formulas of the type $\cond{\bot}{\alpha}_i$. 
The probability associated with such formulas is non-zero if and only if $\alpha$ is impossible in all the successor states,
so in these states, we are able to substitute $\alpha$ with $\bot$.

Finally axioms {\bf A3} and {\bf A4} allow us to eliminate conditional expectation operators.

As with the aleatoric calculus, soundness can be shown by deriving equivalence of the semantic evaluation equations, 
although the proofs are more complex.

We show the soundness of {\bf A0} and {\bf A1}. The soundness of {\bf A2-4} follows immediately from the semantic definition.

\begin{lemma}\label{lem:A0Sound}
  The Axiom {\bf A0} is sound.
\end{lemma}
\begin{proof}
We derive soundness for {\bf A0} from right to left as follows, assuming an arbitrary probability model $P = (W,\pi,f)$:
  We use the abbreviation $E_w^i(\alpha)$ from Definition~\ref{def:semantics} to allow a succinct presentation.
  \begin{eqnarray*}
  && P_w(\ite{\cond{x}{a}_i}{\cond{y}{x\land a}_i}{\cond{z}{\lnot x\land a}_i})\\ 
  &=&
    P_w(\cond{x}{a}_i).P_w(\cond{y}{x\land a}_i)+(1-P_w(\cond{x}{a}_i)).P_w(\cond{z}{\lnot x\land a}_i).
  \end{eqnarray*}
  For the first term we have:
  \begin{eqnarray*}
  && P_w(\cond{x}{a}_i).P_w(\cond{y}{x\land a}_i)\\
  &=&
    \frac{E_w^i(x\land a)}{E_w^i(a)}.
    \frac{E_w^i(y\land x\land a)}{E_w^i(x\land a)}\\
    &=&\frac{E_w^i(y\land x\land a)}{E_w^i(a)}.
  \end{eqnarray*}
  For the second term we have:
  \begin{eqnarray*}
  &&(1-P_w(\cond{x}{a}_i)).P_w(\cond{z}{\lnot x\land a}_i)\\
  &=&
    \frac{E_w^i(a)- E_w^i(x\land a)}{E_w^i(a)}.
    \frac{E_w^i(z\land a\land\lnot x)}{E_w^i(a\land\lnot x)}\\
  &=&
    \frac{E_w^i(a)- E_w^i(x\land a)}{E_w^i(a)}.
    \frac{E_w^i(z\land a\land\lnot x)}{E_W^i(a)-E_w^i(a\land x)}\\
  &=&
    \frac{E_w^i(z\land a\land\lnot x)}{E_w^i(a)}
  \end{eqnarray*}
  Putting these together we get:
  \begin{eqnarray*}
  &&P_w(\ite{\cond{x}{a}_i}{\cond{y}{x\land a}_i}{\cond{z}{\lnot x\land a}_i})\\
  &=&
    \frac{E_w^i(z\land a\land\lnot x)+E_w^i(y\land x\land a)}{E_w^i(a)}\\
  &=&
    \frac{E_w^i(\ite{x}{y}{z})}{E_w^i(a)}\\
  &=&P_w(\cond{\ite{x}{y}{z}}{a}_i)
\end{eqnarray*}
\end{proof}

\begin{lemma}\label{lem:A1Sound}
  The Axiom {\bf A1} is sound.
\end{lemma}
\begin{proof}
  We can show that {\bf A1} is sound by considering cases:
  \begin{itemize}
     \item If $E_w^i(x) = 0$ and $E_w^i(y) = 0$, then $E_w^i(x\lor y) = 0$, so 
       $P_w(\cond{bot}{x}_i) = 1$, $P_w(\cond{x}{y}_i) = 1$ and $P_w(\cond{\bot}{x\lor y}_i) = 1$.
     \item If $E_w^i(x) = 0$ and $E_w^i(y) \neq 0$, then $E_w^i(x\lor y) \neq 0$, so 
       $P_w(\cond{x}{y}_i) = 0$ and $P_w(\cond{\bot}{x\lor y}_i) = 0$.
     \item If $E_w^i(x) \neq 0 $ and $E_w^i(y) = 0$, then $E_w^i(x\lor y) \neq 0$, so 
       $P_w(\cond{bot}{x}_i) = 0$ and $P_w(\cond{\bot}{x\lor y}) = 0$.
     \item If $E_w^i(x) \neq 0$ and $E_w^i(y) \neq 0$, then $E_w^i(x\lor y) \neq 0$, so 
       $P_w(\cond{bot}{x}_i) = 0$ and $P_w(\cond{\bot}{x\lor y}) = 0$.
  \end{itemize}
  In all cases the equality holds, so the axiom is sound.
\end{proof}

\begin{corollary}
  The system \macAx\ is sound for \mac.
\end{corollary}
\begin{proof}
  The nontrivial cases of {\bf A0} and {\bf A1} are presented above and the axioms 
  {\bf A2}, {\bf A3} and {\bf A4} follow immediately from the semantics.
\end{proof}

We conjecture that \macAx\ is complete for the given semantics.
\begin{conjecture}
  The system \macAx\ is complete for \mac.
\end{conjecture}
The intuition for this conjecture is based on the correspondence between \mac\ and modal logic.
This correspondence is formalised in the following section. 
The axiomatization \macAx\ is adequate to show all basic equivalences of the calculus, 
and mirror modal reasoning in {\bf K}.
We anticipate that completeness can be shown in a similar fashion to the completeness of \acAx\ above: 
\begin{enumerate}
  \item we define a cannonical form where non-equivalent cannonical forms have a witnessing structure,
    that gives a different interpretation of each formula;
  \item we show that every formula is provably equivalent to a cannonical formula in \macAx.
\end{enumerate}
This is sufficien to show completeness, but so far the required constructions have been hard to find.

\section{Expressivity}
\newcommand{\KDn}{\ensuremath{\mathrm{K}_n}}

In this section we show that the modal aleatoric calculus generalises the modal logic $K$.
The syntax and semantics of modal logic are given over a set of atomic propositions $Q$.
The syntax of \KDn\ is given by:
$$\phi::=\ q\ |\ \phi\land\phi\ |\ \lnot\phi\ |\ \Box\phi$$
where $q\in Q$, and the operators are respectively {\em and}, {\em not}, {\em necessary}.
The semantics of \KDn\ are given with respect to an epistemic model $M = (W,R,V)$ where $W$ is the nonempty, countable set of possible worlds, 
$R\subseteq W\times W$ is the accessibility relation, and $V:Q\longrightarrow 2^W$ is an assignment of propositions to states.
We require that:
$$\begin{array}{cl}
  \textrm{1}& \forall w,u,v\in W,\ u\in R(w)\ \textrm{and}\ v\in R(w)\ \textrm{implies}\ v\in R(u)\\
  \textrm{2}& \forall w,u,v\in W,\ u\in R(w)\ \textrm{and}\ v\in R(u)\ \textrm{implies}\ v\in R(w)\\
  \textrm{3}& \forall w\in W,\ R(w)\neq\emptyset.
\end{array}$$
We describe the set of worlds $\|\alpha\|^M$ in the model $M = (W,R,V)$ that satisfy the formula $\alpha$ by induction as follows:
$$\begin{array}{lcl}
  \|q\|^M = V(q) &\ & 
    \|\alpha\land\beta\|^M = \|\alpha\|^M\cap\|\beta\|^M\\
  \|\lnot]\alpha\|^M = W-\|\alpha\|^M &\ & 
    \|\Box\alpha\|^M = \{u\in W\ |\ uR\subseteq\|\alpha\|^M\}\ \\
\end{array}$$
where $\forall u\in W$, $uR^\alpha = u^R\cap\|\alpha\|^M$, if $uR^\alpha\cap\|\alpha\|^M\neq\emptyset$ and $uR^\alpha = uR$, otherwise. 

We say \mac\ {\em generalises} $K$ if there is some map $\Lambda$ from pointed epistemic models to pointed probability models, 
and some map $\lambda$ from $K$ formulae to \mac\ formulae such that for all pointed epistemic models $M_w$, for all $K$ formulae $\phi$,
$w\in\|\phi\|^M$ if and only if $\Lambda(M_w)(\lambda(\phi)) = 1$.

We suppose that for every atomic proposition $q\in Q$, there is a unique atomic variable $x_q\in X$. Then the map $\Lambda$ is defined as follows:
Given $M = (W,R,V)$ and $w\in W$, $\Lambda(M_w) = P_w$ where $P = (W,\pi,f)$ and 
\begin{itemize}
  \item $\forall u,v\in W$, $\pi_i(u,v)>0$ if and only if $v\in uR$\footnote{We note this function is not deterministic, but this does not impact the final result.}.
  \item $\forall w\in W$, $\forall q\in Q$, $f_w(x_q)=1$ if $w\in V(q)$ and $f_w(x_q) = 0$ otherwise.
\end{itemize}
This transformation replaces the atomic propositions with variables that, at each world, are either always true or always false, 
and replaces the accessibility relation at a world $w$ with a probability distribution that is non-zero for precisely the worlds accessible from $w$.
It is clear that there is a valid probability model that satisfies these properties.

We also define the map $\lambda$ from $K$ to \mac\ with the following induction:
$$\begin{array}{rclcrcl}
  \lambda(q) &=& x_q &\qquad& \lambda(\alpha\land\beta)&=&\ite{\lambda(\alpha)}{\lambda(\beta)}{\bot}\\
  \lambda(\lnot\alpha)&=& \ite{\lambda(\alpha)}{\bot}{\top} &\qquad& \lambda(\Box\alpha) &=& \cond{\bot}{\ite{\lambda(\alpha)}{\bot}{\top}}
\end{array}$$

\begin{lemma}
  For all epistemic models $M = (W,R,V)$, for all $w\in W$, for all $K$ formula $\phi$, we have $w\in \|\phi\|^M$ if and only if $\Lambda(M_w)(\lambda(\phi)) = 1$.
\end{lemma}
\begin{proof}
  Given $M = (W,R,V)$ and $w\in W$, we let $P_w = \Lambda(M_w)$, where $P = (W,pi,f)$.
  We proceed by induction over the complexity of the \KDn\ formula $\phi$. 
  The induction hypothesis is that for all subformulae $\psi$ of $\phi$, for all $w\in W$, $w\in\|\psi\|^M$ implies $P_w(\psi) = 1$, and $w\notin\|\psi\|^M$ implies $P_w(\psi) = 0$.
  \begin{itemize}
    \item As the base case of the induction, we suppose $\phi = q$. Then $w\in\|\phi\|^M$ implies $w\in V(q)$ which implies $f(x_q)=1$. 
      Conversely, if $w\notin\|\phi\|^M$, then $w\notin V(q)$ so $f(x_q) = 0$.
    \item Given $\phi = \psi_1\land\psi_2$, if $w\in\|\phi\|^M$, then $w\in\|\psi_1\|\cap\|\psi_2\|$, so by the induction hypothesis,
      $P_w(\lambda(\psi)) = P_w(\lambda(\psi_2))= 1$. 
      As $\lambda(\phi) = \ite{\lambda(\psi_1)}{\lambda(\psi_2)}{\bot}$ we have $P_w(\lambda(\phi)) = P_w(\lambda(\psi_1)).P_w(\lambda(\psi_2) = 1$.
      Conversely, if $w\notin\|\psi\|^M$, then either $w\notin\|\psi_1\|^M$, or $w\notin\|\psi_2\|^M$. 
      We assume, w.l.o.g., $w\notin\|\psi_1\|$, and by the induction hypothesis, $P_w(\lambda(\psi_1) = 0$.
      Thus, $P_w(\lambda(\phi)) = P_w(\lambda(\psi_1)).P_w(\lambda(\psi_2) = 0$.
    \item Given $\phi = \lnot\psi$, we have $w\in\|\phi\|^M$ if and only if $w\notin\|\psi\|^M$. 
      If $w\in\|\psi\|^M$, then $P_w(\lambda(\psi)) = 1$, so $P_w(\ite{\lambda(\psi)}{\bot}{\top}) = 0$.
      If $w\notin\|\psi\|^M$, then $P_w(\lambda(\psi)) = 0$, so we have $P_w(\lambda(\phi)) = 1$.
    \item Given $\phi = \Box_i\psi$, then $w\in \|\phi\|^M$ if and only if for all $u\in wR_i$, $u\in\|\psi\|^M$.
      By the induction hypothesis we have for all $u\in wR_i$, $P_u(\psi) = 1$, so $P_u(\ite{\lambda(\psi)}{\bot}{\top}) = 0$.
      From the definition of $\Lambda$ we have $u\in wR_i$ if and only if $\pi_i(w,u)\neq 0$.
      Therefore, $\sum_{u\in W}\pi_i(w,u)P_u(\ite{\lambda(\psi)}{\bot}{\top}) = 0$ and
      so by Definition~\ref{semantics}, $P_w(\lambda(\phi)) = P_w(\cond{\bot}{\ite{\lambda(\psi)}{\bot}{\top}}_i) = 1$.
      Conversely, if $w\notin\|\phi\|^M$, then for some $u\in wR_i$ we have $u\notin\|\psi\|^M$. 
      By the induction hypothesis, $P_u(\lambda(\psi)) = 0$, and as $\pi_i(w,u)>0$, we have $\sum_{u\in W}\pi_i(w,u)P_u(\ite{\lambda(\psi)}{\bot}{\top}) > 0$.
      Therefore, the $P_u(\bot)$ term dominates the expression for $P_w(\lambda(\phi))$ as follows:
      $$P_w(\lambda(\phi)) = \frac{\sum_{u\in W}\pi_i(w,u).P_u(\ite{\lambda(\psi)}{\bot}{\top}).P_u(\bot)}{\sum_{u\in W}\pi_i(w,u)P_u(\ite{\lambda(\psi)}{\bot}{\top})} = 0.$$
  \end{itemize}
\end{proof}

\section{Case Study}\label{sect:case}
We present a case study using some simple actions in a dice game illustrating the potential for reasoning in AI applications.
A simple version of the game {\em pig}\footnote{https://en.wikipedia.org/wiki/Pig\_(dice\_game)} uses a four sided dice, 
and players take turns. Each turn, the player rolls the dice as many times as they like, adding the numbers the roll to their turn total. 
However, if they roll a 1, their turn total is set to 0, and their turn ends.
They can elect to stop at any time, in which case their turn total is added to their score.

\newcommand{\odd}{\ensuremath{\mathtt{odd}}}
\newcommand{\gtt}{\ensuremath{\mathtt{gt2}}}
\newcommand{\risk}{\ensuremath{\mathtt{risk}}}

To illustrate the aleatoric calculus we suppose that for our dice we have two random variables, {\tt odd} and {\tt gt2} ({\em greater than 2}).
Every roll of the dice can be seen as a sampling of these two variables: 1 is an odd number not greater than 2, and so on.
Now we suppose that there is some uncertainty to the fairness of the dice, so it is possible that there is a 70\% chance of the dice rolling a number greater than 2.
However, we consider this unlikely and only attach a 10\% likelihood to this scenario.
Finally, we suppose there is an additional random variable called \risk\ which can be used to define a policy. 
For example, we might roll again if the \risk\ variable is sampled as true.
This scenario if visualised in Figure~\ref{pig}.
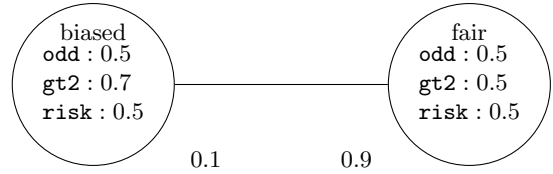
\begin{figure}
  \begin{center}
    \scalebox{1.0}{
  \begin{tikzpicture}
    \draw (0,0) node[circle,draw](loaded) {$\begin{array}{l}\odd:0.5\\ \gtt:0.7\\ \risk:0.5\end{array}$};

    \draw (0,0.7) node {biased};
    \draw (5,0) node[circle,draw](fair) {$\begin{array}{l}\odd:0.5\\\gtt:0.5\\\risk:0.5\end{array}$};
      \draw (5,0.7) node {fair};
    \draw (loaded)--(fair);
    \draw (1.5,-1) node {$0.1$};
    \draw (3.5,-1) node {$0.9$};
  \end{tikzpicture}
}
  \caption{A simple two world representation of the game pig, where the dice is possibly biased. 
    The biased world is the actual world.}\label{pig}
  \end{center}
\end{figure}

We can now build aleatoric formula describing various situations, assuming the dice is actually biased:

\noindent
\begin{tabular}{|l|l|l|c|}
  \hline
  Name & Formula & Description & Prob\\
  \hline
  {\tt bust} & $\ite{\gtt}{\bot}{\odd}$ & the probability of rolling a 1 & 0.15\\
  \hline
  {\tt four} & $\ite{\odd}{\bot}{\gtt}$ & the probability of rolling a 4 &0.35\\
  \hline
  {\tt thinkBust} & $\cond{\mathtt{bust}}{\top}$ & the chance given to rolling a 1& 0.265\\
  \hline
  {\tt think-4-1} & $\cond{\mathtt{bust}}{\mathtt{four}}$ & the chance of rolling a 1 given a 4  & 0.237\\
  \hline
  {\tt rollAgain} & $\ite{{\tt thinkBust}^{1/2}}{\mathtt{risk}}{\top}$ & whether to roll again & 0.77\\
  \hline
\end{tabular}

These formulas show the different types of information that can be represented: 
{\tt bust} and {\tt four} are true random variables (aleatoric information), 
whereas {\tt thinkBust} and {\tt think-4-1} are based on an agent's mental model (Bayesian information).
Finally {\tt rollAgain} describes the condition for a policy. 
In a dynamic extension of this calculus, given prior assumptions about policies, 
agents may apply Bayesian conditioning to learn probability distributions from observations.

\section{Conclusion}\label{sect:conclusion}

The modal aleatoric calculus is shown to be a true generalisation of modal logic, 
but gives a much richer language that encapsulates probabilistic reasoning and degrees of belief.
We have shown that the modal aleatoric calculus is able to describe probabilistic strategies 
for agents. 
We have provided a sound axiomatization for the calculus, shown it is complete for the aleatoric calculus
and we are working to show that the axiomatization is complete for the modal aleatoric calculus.
Future work will consider dynamic variations of the logic, 
where agents apply Bayesian conditioning based on their observations to learn the probability distribution of worlds.

\providecommand{\noopsort}[1]{}


\begin{thebibliography}{10}
\providecommand{\url}[1]{\texttt{#1}}
\providecommand{\urlprefix}{URL }
\providecommand{\doi}[1]{https://doi.org/#1}

\bibitem{baltagetal.tark:2009}
Baltag, A., Smets, S.: Group belief dynamics under iterated revision: fixed
  points and cycles of joint upgrades. In: Proc.\ of 12th TARK. pp. 41--50
  (2009)

\bibitem{BayesEp}
Bovens, L., Hartmann, S.: Bayesian epistemology. Oxford University Press (2003)

\bibitem{tensorlog}
Cohen, W.W., Yang, F., Mazaitis, K.R.: Tensorlog: Deep learning meets
  probabilistic dbs. arXiv preprint arXiv:1707.05390  (2017)

\bibitem{deFinetti}
De~Finetti, B.: Theory of probability: a critical introductory treatment. John
  Wiley \& Sons (1970)

\bibitem{hvdetal.del:2007}
{\noopsort{Ditmarsch}}{van Ditmarsch}, H., {\noopsort{Hoek}}{van der Hoek}, W.,
  Kooi, B.: Dynamic Epistemic Logic, Synthese Library, vol.~337. Springer
  (2007)

\bibitem{faginHalpern}
Fagin, R., Halpern, J.Y., Megiddo, N.: A logic for reasoning about
  probabilities. Information and computation  \textbf{87}(1-2),  78--128 (1990)

\bibitem{feldman}
Feldman, Y.A., Harel, D.: A probabilistic dynamic logic. In: Proceedings of the
  fourteenth annual ACM symposium on theory of computing. pp. 181--195 (1982)

\bibitem{hailperin}
Hailperin, T.: Boole's logic and probability: a critical exposition from the
  standpoint of contemporary algebra, logic and probability theory  (1976)

\bibitem{halpern}
Halpern, J.Y.: Reasoning about uncertainty. MIT press (2017)

\bibitem{hintikka}
Hintikka, J.: Knowledge and Belief. Cornell University Press (1962)

\bibitem{hommerson}
Hommersom, A., Lucas, P.J.: Generalising the interaction rules in probabilistic
  logic. In: 22nd International Joint Conference on Artificial Intelligence
  (2011)

\bibitem{jfaketal.prob:2009}
J.~{\noopsort{Benthem}}{van Benthem}, J.G., Kooi, B.: Dynamic update with
  probabilities. Studia Logica  \textbf{93(1)},  67--96 (2009)

\bibitem{kooi}
Kooi, B.P.: Probabilistic dynamic epistemic logic. Journal of Logic, Language
  and Information  \textbf{12}(4),  381--408 (2003)

\bibitem{kooi.ll:2011}
Kooi, B.: Dynamic epistemic logic. In: {\noopsort{Benthem}}{van Benthem}, J.,
  ter Meulen, A. (eds.) Handbook of Logic and Language. pp. 671--690. Elsevier
  (2011), second edition

\bibitem{kozen}
Kozen, D.: A probabilistic {PDL}. Journal of Computer and System Sciences
  \textbf{30}(2),  162--178 (1985)

\bibitem{Milne}
Milne, P.: Bruno de {F}inetti and the logic of conditional events. The British
  Journal for the Philosophy of Science  \textbf{48}(2),  195--232 (1997)

\bibitem{nilsson}
Nilsson, N.J.: Probabilistic logic. Artificial intelligence  \textbf{28}(1),
  71--87 (1986)

\bibitem{plaza:1989}
Plaza, J.: Logics of public communications. In: Proc.\ of the 4th ISMIS. pp.
  201--216. Oak Ridge National Laboratory (1989)

\bibitem{stalnaker}
Stalnaker, R.C.: A theory of conditionals. In: Ifs, pp. 41--55. Springer (1968)

\bibitem{stalnakerThomason}
Stalnaker, R.C., Thomason, R.H.: A semantic analysis of conditional logic 1.
  Theoria  \textbf{36}(1),  23--42 (1970)

\bibitem{williamson}
Williamson, J.: From bayesian epistemology to inductive logic. Journal of
  Applied Logic  \textbf{11}(4),  468--486 (2013)

\bibitem{zadeh}
Zadeh, L.A.: Fuzzy sets. In: Fuzzy Sets, Fuzzy Logic, And Fuzzy Systems:
  Selected Papers by Lotfi A Zadeh, pp. 394--432. World Scientific (1996)

\end{thebibliography}
\end{document}